\documentclass[envcountsame,fleqn]{llncs}

\pdfoutput=1

\usepackage[T1]{fontenc}
\usepackage{lmodern}
\usepackage{mathtools}
\usepackage{amssymb}
\usepackage{xspace}
\usepackage{booktabs}
\usepackage[basic,small]{complexity}
\usepackage[compatibility=true]{caption}
\usepackage{floatrow}
\usepackage{tikz}
\usepackage{xspace}
\usepackage[final]{microtype}
\usepackage[breaklinks,%
  bookmarks=false,%
  pdfstartview={FitH},%
  pdfborder={0 0 0}]%
{hyperref}
\usepackage[capitalise]{cleveref}

\newif\ifapdx
\apdxtrue

\InputIfFileExists{glyphtounicode}{\pdfgentounicode=1}{}

\DeclareNewFloatType{algorithm}{name={Algorithm}}
\crefname{algorithm}{Algorithm}{Algorithms}

\makeatletter
\def\@spthm#1#2#3#4{\topsep 7\p@ \@plus2\p@ \@minus4\p@
\refstepcounter[#2]{#1}%
\@ifnextchar[{\@spythm{#1}{#2}{#3}{#4}}{\@spxthm{#1}{#2}{#3}{#4}}}
\makeatother

\spnewtheorem*{claim*}{Claim}{\itshape}{\rmfamily}
\spnewtheorem*{lemma*}{Lemma}{\upshape\bfseries}{\itshape}
\spnewtheorem*{theorem*}{Theorem}{\upshape\bfseries}{\itshape}
\spnewtheorem*{proposition*}{Proposition}{\upshape\bfseries}{\itshape}
\spnewtheorem*{corollary*}{Corollary}{\upshape\bfseries}{\itshape}

\crefname{Definition}{Definition}{Definitions}
\crefname{Proposition}{Proposition}{Propositions}
\crefname{Lemma}{Lemma}{Lemmas}
\crefname{Theorem}{Theorem}{Theorems}
\crefname{Corollary}{Corollary}{Corollaries}
\crefname{Claim}{Claim}{Claims}
\crefname{Example}{Example}{Examples}
\crefname{Remark}{Remark}{Remarks}
\crefname{Exercise}{Exercise}{Exercises}
\crefname{Problem}{Problem}{Problems}
\crefname{Note}{Note}{Notes}
\crefname{Question}{Question}{Questions}
\crefname{Solution}{Solution}{Solutions}

\crefname{section}{Sect.}{Sects.}
\Crefname{section}{Section}{Sections}
\crefname{subsection}{Sect.}{Sects.}
\Crefname{section}{Section}{Sections}
\crefname{figure}{Fig.}{Figs.}
\Crefname{figure}{Figure}{Figures}

\crefformat{equation}{(#2#1#3)}
\Crefformat{equation}{Eq.~(#2#1#3)}
\crefrangeformat{equation}{(#3#1#4)--(#5#2#6)}
\Crefrangeformat{equation}{Eqs.~(#3#1#4)--(#5#2#6)}
\crefmultiformat{equation}{(#2#1#3)}{ and~(#2#1#3)}{, (#2#1#3)}{ and~(#2#1#3)}
\Crefmultiformat{equation}{Eqs.~(#2#1#3)}{ and~(#2#1#3)}{, (#2#1#3)}{ and~(#2#1#3)}
\crefrangemultiformat{equation}{(#3#1#4)--(#5#2#6)}{ and (#3#1#4)--(#5#2#6)}{, (#3#1#4)--(#5#2#6)}{ and (#3#1#4)--(#5#2#6)}
\Crefrangemultiformat{equation}{Eqs.~(#3#1#4)--(#5#2#6)}{ and (#3#1#4)--(#5#2#6)}{, (#3#1#4)--(#5#2#6)}{ and (#3#1#4)--(#5#2#6)}

\let\phi\varphi
\let\epsilon\varepsilon
\let\theta\vartheta

\let\assign\leftarrow

\newcommand{\abs}[1]{\lvert #1\rvert}
\newcommand{\size}[1]{\lVert #1\rVert}
\newcommand{\dotminus}{\mathbin{\dot{\smash{-}\vphantom{a}}}}
\newcommand{\bbN}{\ensuremath{\mathbb{N}}}

\newcommand{\bbR}{\ensuremath{\mathbb{R}}}
\newcommand{\calM}{\ensuremath{\mathcal{M}}}

\DeclareMathOperator{\Prob}{Pr}

\DeclareMathOperator*{\opt}{opt}
\DeclareMathOperator{\val}{val}
\DeclareMathOperator{\Oh}{O}

\newclass{\ExpTime}{EXPTIME}
\newclass{\TwoExpTime}{2EXPTIME}

\let\P\undefined
\newcommand{\CTL}{\ensuremath{\mathrm{CTL}}\xspace}

\newcommand{\PCTL}{\ensuremath{\mathrm{PCTL}}\xspace}
\newcommand{\PRCTL}{\ensuremath{\mathrm{PRCTL}}\xspace}
\newcommand{\True}{\ensuremath{\mathit{true}}}
\newcommand{\P}{\mathord{\mathsf{P}}}
\newcommand{\AllP}{\mathord{\forall\mathsf{P}}}
\newcommand{\ExP}{\mathord{\exists\mathsf{P}}}
\newcommand{\Until}{\mathbin{\mathsf{U}}}
\newcommand{\Release}{\mathbin{\mathsf{R}}}
\newcommand{\Finally}{\mathord{\mathsf{F}}}
\newcommand{\Globally}{\mathord{\mathsf{G}}}
\newcommand{\Next}{\mathord{\mathsf{X}}}

\newcommand{\e}{\mathrm{e}}

\let\Labels\lambda
\newcommand{\reward}{\mathit{rew}}
\newcommand{\States}{S}

\newcommand{\Actions}{\mathit{Act}}

\newcommand{\ie}{i.e.\@\xspace}
\newcommand{\eg}{e.g.\@\xspace}
\newcommand{\wrt}{wrt.\@\xspace}
\newcommand{\ea}{et al.\@\xspace}

\usetikzlibrary{arrows}
\usetikzlibrary{shapes}
\tikzset{every picture/.style={>=stealth,bend angle=20,label distance=1pt}}
\tikzset{every label/.style={font=\small}}
\tikzset{every node/.style={inner sep=2pt,font=\small}}
\tikzset{state/.style={circle split,draw,minimum size=5ex,font=\footnotesize}}

\title{Computing Quantiles in\\
Markov Reward Models\thanks{This work was
supported by the DFG project QuaOS and
the collaborative research centre HAEC (SFB 912) funded by the DFG.
This work was partly supported by
the European Union Seventh Framework Programme under 
grant agreement no.~295261 (MEALS), the DFG/NWO project ROCKS
and the cluster of excellence cfAED.}}
\author{Michael Ummels\inst{1} \and Christel Baier\inst{2}}
\institute{Institute of Transportation Systems, German Aerospace Center\\
\email{michael.ummels@dlr.de} \and
Technische Universit\"at Dresden\\
\email{baier@tcs.inf.tu-dresden.de}}

\begin{document}
\maketitle
\setcounter{footnote}{0}

\begin{abstract}
Probabilistic model checking mainly concentrates on techniques
for reasoning about the probabilities of certain path properties
or expected values of certain random variables.
%
For the quantitative system analysis, however, there is also another type 
of interesting performance measure, namely \emph{quantiles}.
A typical quantile query takes as input a lower probability bound 
$p \in {]0,1]}$ and a reachability property. 
The task is then to compute the minimal reward bound~$r$ 
such that with probability at least~$p$
the target set will be reached before the accumulated reward
exceeds~$r$.
Quantiles are well-known from mathematical statistics, but to the best
of our knowledge they have not been addressed by the model checking
community so far. 

In this paper, we 
study the complexity of quantile queries for until properties in
discrete-time finite-state Markov decision processes with 
nonnegative rewards on states.
We show that qualitative quantile queries can be evaluated
in polynomial time and 
present an exponential algorithm for the evaluation of
quantitative quantile queries.
For the special case of Markov chains, we show that
quantitative quantile queries can be evaluated in
pseudo-polynomial time.
\end{abstract}

\section{Introduction}

Markov models with reward (or cost) functions are widely used
for the quantitative system analysis.
We focus here on the discrete-time or time-abstract case. 
Discrete-time Markov decision processes, MDPs for short,
can be used, for instance, as an operational model for
randomised distributed algorithms and rewards might serve
to reason, \eg, about the size of the buffer of a communication
channel or
about the number of rounds that a leader election
protocol might take until a leader has been elected.

Several authors considered variants of probabilistic computation tree
logic (\PCTL) \cite{HanJon94,BdAlf95}
for specifying quantitative constraints on the behaviour
of Markov models with reward functions.
Such extensions, briefly called \PRCTL here, permit to specify constraints
on the probabilities of reward-bounded reachability conditions,
on the expected accumulated rewards until a certain set 
of target states is reached 
or expected instantaneous rewards after some fixed number of steps
\cite{AlfaroSTACS97,AlfaroThesis97,AlfaroCONCUR99,AndovaHK03,PekYounes05},
or on long-run averages \cite{AlfaroLICS98}.
An example for a typical \PRCTL formula with \PCTL's probability
operator and the reward-bounded until operator
is the formula
$\P_{>p}(a \Until_{\leq r} b)$ where
$p$~is a lower probability bound in $[0,1[$ and 
$r$~is an upper bound for the accumulated 
reward earned by path fragments that lead via states where $a$~holds to a
$b$-state. 
From a practical point of view, more important than 
checking whether a given \PRCTL formula~$\phi$ holds for
(the initial state of) a Markov model~$\calM$
are \PRCTL \emph{queries} of the form $\P_{=?}\,\psi$  
where the task is to calculate the (minimum or maximum)
probability for the path formula~$\psi$.
%
Indeed, the standard \PRCTL model checking algorithm checks whether a given
formula $\P_{\bowtie\mkern 1mu p}\,\psi$ holds in~$\calM$ by 
evaluating the \PRCTL query $\P_{=?}\,\psi$ and 
comparing the computed value~$q$ with the given 
probability bound~$p$ according to the comparison
predicate~$\bowtie$. 
The standard procedure for dealing with \PRCTL formulas that refer to
expected (instantaneous or accumulated) rewards relies on an 
analogous scheme; see \eg \cite{FKNP11}.
An exception can be made for qualitative \PRCTL properties 
$\P_{\bowtie\mkern 1mu p}\,\psi$ where the probability bound~$p$
is either $0$ or~$1$, and the  path formula~$\psi$
is a plain until formula without reward bound (or any 
$\omega$-regular path property without reward constraints):
in this case, a graph analysis suffices to check whether
$\P_{\bowtie\mkern 1mu p}\,\psi$ holds for~$\calM$
\cite{VardiFOCS85,CouYanJACM95}.

In a common project with the operating system group of our department, 
we learned that a natural question for the systems community
is to swap the given and unknown parameters
in \PRCTL queries and to ask for the computation of a \emph{quantile}
(see \cite{BDEHKKMTV12}).
For instance, if $\calM$~models a mutual exclusion protocol for
competing processes $P_1,\ldots,P_n$
and rewards are used to represent the time spent by process~$P_i$ in its
waiting location, then the \emph{quantile query} 
$\P_{>0.9}(\mathit{wait}_i\Until_{\leq ?}\mathit{crit}_i)$ 
asks for the minimal time bound~$r$ such that 
in all scenarios (i.e., under all schedulers)
with probability greater than 0.9
process~$P_i$ will wait  no longer than $r$~time units 
before entering its critical section.
For another example, suppose $\calM$~models the management
system of a service execution platform. Then the query
$\P_{>0.98}(\True\Until_{\leq ?}\mathit{tasks\_completed})$ 
might ask for the minimal
initial energy budget~$r$ that is required to ensure that 
even in the worst-case
there is more than 98\% chance to 
reach a state where all tasks have been completed successfully.

To the best of our knowledge, quantile queries have not yet
been addressed directly in the model checking community.
What is known from the literature is that for finite
Markov chains with nonnegative rewards
the task of checking whether a \PRCTL formula
$\P_{>p}(a\Until_{\leq r} b)$ or
$\P_{\geq p}(a\Until_{\leq r} b)$
holds for some given state is \NP-hard \cite{LaroussinieS05} 
when $p$ and $r$~are represented in
binary.
Since such a formula holds in state~$s$ if and only if
the value of the corresponding quantile query at~$s$
is $\leq r$,
this implies that evaluating quantile queries is also
\NP-hard.

The purpose of this paper is to study quantile queries for
Markov decision processes with nonnegative rewards in more details. 
We consider quantile queries for reward-bounded until formulas
in combination with the standard \PRCTL quantifier
$\P_{\bowtie\mkern 1mu p}$ (in this paper denoted by
$\AllP_{\bowtie\mkern 1mu p}$),
where universal quantification over all schedulers
is inherent in the semantics, and 
its dual $\ExP_{\bowtie\mkern 1mu p}$ that asks for the existence of some
scheduler enjoying a certain property. 
By duality, our results carry over to reward-bounded release properties.

\subsubsection{Contributions.}

First, we address \emph{qualitative} quantile queries,
\ie quantile queries where the probability bound is
either $0$ or~$1$, and we show that such queries can
be evaluated in strongly polynomial time. 
Our algorithm is surprisingly simple and does not
rely on value iteration or linear programming
techniques
(as it is e.g. the case for extremal expected reachability times
and stochastic shortest-paths problems in MDPs \cite{AlfaroCONCUR99}).
Instead, our algorithm relies on the greedy method and borrows ideas 
from D\ij{}kstra's shortest-path algorithm.
In particular, our algorithm can be used for checking PRCTL
formulas of the form $\AllP_{\bowtie\mkern 1mu p}(a\Until_{\leq r} b)$
or $\ExP_{\bowtie\mkern 1mu p}(a\Until_{\leq r} b)$
with $p\in\{0,1\}$ in polynomial time. Previously,
a polynomial-time algorithm was known only for the special case
of MDPs where every loop contains a state with nonzero reward
\cite{JurdzinskiSL08}.
%
%

Second, we consider \emph{quantitative} quantile queries.
The standard way to compute
the maximal or minimal probabilities for reward-bounded until
properties,  say $a \Until_{\leq r} b$,
relies on the iterative computation of the extremal probabilities
$a \Until_{\leq i} b$ for increasing reward bound~$i$.
We use here a reformulation of this computation scheme
as a linear program whose size is polynomial in the number of states 
of~$\calM$ and the given reward bound~$r$.
The crux to derive from this linear program an algorithm 
for the evaluation of quantile queries is to provide a bound for the 
sought value, which is our second contribution.
This bound then permits to perform a sequential
search for the
quantile, which yields an exponentially time-bounded algorithm for
evaluating quantitative quantile queries.
Finally, in the special case of Markov chains with integer rewards,
we show that this algorithm can be improved to run in time polynomial
in the size of the query, the size of the chain, and the largest
reward, \ie in \emph{pseudo-polynomial} time. 


\subsubsection*{Outline.}
The structure of the paper is as follows.
\Cref{section:prelim} summarises the relevant
concepts of Markov decision processes and briefly recalls
the logic \PRCTL.
Quantile queries are introduced in \cref{section:queries}.
Our polynomial-time algorithms for qualitative quantile queries
is presented in \cref{section:qualitative}, whereas
the quantitative case is addressed in 
\cref{section:quantitative}.
The paper ends with some concluding remarks in 
\cref{section:conclusion}.

\section{Preliminaries}
\label{section:prelim}

In the following, we assume a countably infinite
set \AP\ of \emph{atomic propositions}.
A Markov decision process (MDP)
$\calM=(\States,\Actions,\gamma,\Labels,\reward,\delta)$
with nonnegative rewards consists of a
finite set $\States$ of states,
a finite set $\Actions$ of actions,
a function
$\gamma\colon\States\to 2^\Actions\setminus\{\emptyset\}$
describing the set of \emph{enabled actions} in each state,
a labelling function
$\Labels\colon\States\to 2^\AP$, a reward function
$\reward\colon\States\to\bbR^{\geq 0}$, and
a transition function
$\delta\colon\States\times\Actions\times\States\to[0,1]$
such that $\sum_{t\in\States}\delta(s,\alpha,t)=1$
for all $s\in\States$ and $\alpha\in\Actions$.
If the set $\Actions$ of actions is just a singleton, we
call~$\calM$ a Markov chain.

Given an MDP~$\calM$, we say that a
state~$s$ of~$\calM$ is \emph{absorbing} if
$\delta(s,\alpha,s)=1$ for all $\alpha\in\gamma(s)$.
Moreover, for $a\in\AP$ we denote by $\Labels^{-1}(a)$ the set
of states~$s$ such that $a\in\Labels(s)$, and
for $x=s_0s_1\ldots s_k\in\States^*$ we denote by
$\reward(x)$ the accumulated reward after~$x$, \ie
$\reward(x)=\sum_{i=0}^	k\reward(s_i)$.
Finally, we denote by $\abs{\delta}$ the number of
\emph{nontrivial} transitions in $\calM$, \ie
$\abs{\delta}=\abs{\{(s,\alpha,t):
\text{$\alpha\in\gamma(s)$ and $\delta(s,\alpha,t)>0$}\}}$.

\emph{Schedulers} are used to resolve the nondeterminism that
arises from the possibility that more than one action might
be enabled in a given state. Formally, a scheduler for~$\calM$
is a mapping $\sigma\colon\States^+\to\Actions$ such that
$\sigma(xs)\in\gamma(s)$ for all $x\in\States^*$ and $s\in\States$.
Such a scheduler~$\sigma$ is \emph{memoryless} if
$\sigma(xs)=\sigma(s)$ for all $x\in\States^*$ and $s\in\States$.
Given a scheduler~$\sigma$ and an initial state $s=s_0$,
there is a unique probability measure~$\Prob_s^\sigma\!$
on the Borel $\sigma$-algebra over $\States^\omega$ such that
$\Prob_s^\sigma(s_0 s_1\ldots s_k\cdot S^\omega)
=\prod_{i=0}^{k-1}\delta(s_i,\sigma(s_0\ldots s_i),s_{i+1})$;
see \cite{BaierK08}.

Several logics have been introduced in order to reason
about the probability measures $\Prob_s^\sigma$.
In particular, the logics $\PCTL$ and $\PCTL^*$ replace
the path quantifiers of $\CTL$ and $\CTL^*$ by a single
probabilistic quantifier $\P_{\bowtie\mkern 1mu p}$,
where ${\bowtie}\in\{<,\leq,\geq,>\}$ and $p\in[0,1]$.
In these logics, the formula $\phi=\P_{\bowtie\mkern 1mu p}\,\psi$
holds in state~$s$ (written $s\models\phi$) if under
\emph{all} schedulers~$\sigma$
the probability $\Prob_s^\sigma(\psi)$ of the path property~$\psi$
compares positively with~$p$ \wrt the comparison operator~$\bowtie$,
\ie if $\Prob_s^\sigma(\psi)\bowtie\psi$.
A~dual existential quantifier
$\ExP_{\bowtie\mkern 1mu p}$ that asks for the existence of
a scheduler can be introduced using the equivalence
$\ExP_{\bowtie\mkern 1mu p}\,\psi\equiv
\neg\P_{\overline{\mathord{\bowtie\mkern-1mu}}\mkern 2mu p}
\,\psi$, where $\overline{{\bowtie}}$~denotes the dual
inequality. Since many properties of MDPs can be expressed
more naturally using the $\ExP$ quantifier, we consider this
quantifier an equal citizen of the logic, and we denote
the universal quantifier $\P$ by $\AllP$ in order
to stress its universal semantics.

In order to be
able to reason about accumulated rewards, we amend the
until operator $\Until$ by a reward constraint of the form
$\sim r$, where $\sim$~is a comparison operator and
$r\in\bbR\cup\{\pm\infty\}$.
Since we adopt the convention that a reward is earned
upon \emph{leaving} a state,
a path $\pi=s_0 s_1\ldots$ fulfils the
formula $\psi_1\Until_{\sim r}\psi_2$ if there exists a
point $k\in\bbN$ such that 1.\ $s_k s_{k+1}\ldots\models\psi_2$,
2.\ $s_i s_{i+1}\ldots\models\psi_1$ for all $i<k$, and
3.\ $\reward(s_0\ldots s_{k-1})\sim r$.
Even though our logic is only a subset of the logics
\PRCTL and $\PRCTL^*$ defined in \cite{AndovaHK03}, we use
the same names for the extension of \PCTL and $\PCTL^*$
with the amended until operator. The following proposition
states that extremal probabilities for $\PRCTL^*$ are
attainable. This follows, for instance, from the fact that
$\PRCTL^*$ can only describe $\omega$-regular path properties.

\begin{proposition}\label{prop:optimal-strat}
Let $\calM$ be an MDP and $\psi$ a $\PRCTL^*$ path formula.
Then there exist
schedulers $\sigma^*$ and~$\tau^*$ such that
$\Prob_s^{\sigma^*}\!(\psi)=\sup_\sigma\Prob_s^\sigma(\psi)$
and $\Prob_s^{\tau^*}\!(\psi)=\inf_\tau\Prob_s^\sigma(\psi)$
for all states~$s$ of~$\calM$.
\end{proposition}

\section{Quantile queries}
\label{section:queries}

A \emph{quantile query} is of the form
$\phi=\AllP_{\bowtie\mkern 1mu p}(a\Until_{\leq ?} b)$
or
$\phi=\ExP_{\bowtie\mkern 1mu p}(a\Until_{\leq ?} b)$,
where $a,b\in\AP$, $p\in[0,1]$ and
${\bowtie}\in\{<,\leq,\geq,>\}$.
We call queries of the former type \emph{universal} and
queries of the latter type \emph{existential}.
If $r\in\bbR\cup\{\pm\infty\}$,
we write $\phi[r]$ for the \PRCTL formula that is
obtained from~$\phi$ by replacing $?$ with~$r$.

Given an MDP~$\calM$ with rewards,
evaluating $\phi$ on~$\calM$ amounts to computing, for each
state~$s$ of~$\calM$,
the least or the largest $r\in\bbR$ such that $s\models\phi[r]$.
Formally, if $\phi=\AllP_{\bowtie\mkern 1mu p}(a\Until_{\leq ?} b)$
or $\phi=\ExP_{\bowtie\mkern 1mu p}(a\Until_{\leq ?} b)$
then the \emph{value} of a state~$s$ of~$\calM$ with respect
to~$\phi$ is
$\val^\calM_\phi(s)\coloneqq\opt\{r\in\bbR:s\models\phi[r]\}$,
where $\opt=\inf$ if ${\bowtie}\in\{\geq,>\}$
and $\opt=\sup$ otherwise.\footnote{As usual, we assume that
$\inf\emptyset=\infty$ and $\sup\emptyset=-\infty$.}
Depending on whether $\val^\calM_\phi(s)$ is defined as an infimum or
a supremum, we call $\phi$ a \emph{minimising} or a
\emph{maximising} query, respectively.
In the following, we will omit the superscript~$\calM$ when
the underlying MDP is clear from the context.

Given a query~$\phi$, we define the \emph{dual query}
to be the unique quantile query~$\overline\phi$ such that
$\overline{\phi}[r]\equiv\neg\phi[r]$ for all $r\in\bbR\cup\{\pm\infty\}$.
Hence, to form the dual of a query, one only needs to replace
the quantifier $\AllP_{\bowtie\mkern 1mu p}$ by
$\ExP_{\overline{\mathord{\bowtie\mkern-1mu}}\mkern 2mu p}$ and
vice versa. For instance,
the dual of $\AllP_{<p}(a\Until_{\leq ?} b)$ is
$\ExP_{\geq p}(a\Until_{\leq ?} b)$.
Note that the dual of a universal or minimising query is an
existential or maximising query, respectively, and vice versa.


\begin{proposition}\label{prop:duality}
Let $\calM$ be an MDP and $\phi$ a quantile query.
Then $\val_\phi(s)=\val_{\overline{\phi}}(s)$ for all
states~$s$ of~$\calM$.
\end{proposition}

\begin{proof}
Without loss of generality, assume that $\phi$ is
a minimising query.
Let $s\in\States$, $v=\val_\phi(s)$ and $v'=\val_{\overline\phi}(s)$.
On the one hand, for all $r<v$ we have
$s\not\models\phi[r]$, \ie $s\models\overline\phi[r]$,
and therefore $v'\geq v$.
On the other hand,
since $\phi[r]$ implies $\phi[r']$ for $r'\geq r$,
for all $r>v$ we have
$s\models\phi[r]$, \ie
$s\not\models\overline\phi[r]$,
and therefore also $v'\leq v$.
\qed
\end{proof}

Assume that we have computed the value $\val_\phi(s)$ of a state~$s$
with respect to a quantile query~$\phi$.
Then, for any $r\in\bbR$, to decide whether $s\models\phi[r]$,
we just need to compare~$r$ to~$\val_\phi(s)$.

\begin{proposition}\label{prop:value}
Let $\calM$ be an MDP, $s$ a state of~$\calM$,
$\phi$ a minimising or maximising quantile query, and $r\in\bbR$.
Then $s\models\phi[r]$ if and only if $\val_\phi(s)\leq r$ or
$\val_\phi(s)>r$, respectively.
\end{proposition}

\begin{proof}
First assume that
$\phi=Q(a\Until_{\leq ?} b)$
is a minimizing query.
Clearly, if $s\models\phi[r]$, then $\val_\phi(s)\leq r$.
On the other hand, assume that $\val_\phi(s)\leq r$ and
denote by~$R$ the set of numbers~$x\in\bbR$ of the form
$x=\sum_{i=0}^k\reward(s_i)$ for a finite sequence
$s_0 s_1\ldots s_k$ of states.
Since the set $\{x\in R:x\leq n\}$ is finite for all $n\in\bbN$,
we can fix some $\epsilon>0$ such that
$r+\delta\notin R$ for all $0<\delta\leq\epsilon$.
Hence, the set of paths that fulfil
$a\Until_{\leq r} b$
agrees with the
set of paths that fulfil
$a\Until_{\leq r+\epsilon} b$.
Since $\val_\phi(s)<r+\epsilon$ and $\phi$~is a minimising
query, we know that $s\models\phi[r+\epsilon]$.
Since replacing $r+\epsilon$ by~$r$ does not affect the path property,
this implies that $s\models\phi[r]$.
Finally, if $\phi$~is a maximising query, then $\overline{\phi}$ is
a minimising query, and $s\models\overline{\phi}[r]$ if and only if
$\val_{\overline{\phi}}(s)=\val_\phi(s)\leq r$, \ie
$s\models\phi[r]$ if and only if $\val_\phi(s)>r$.
\qed
\end{proof}

Proposition~\ref{prop:value} does not hold when we allow~$r$ to take
an infinite value. In fact, if $\phi$~is a minimizing query
and $s\not\models\phi[\infty]$, then $\val_\phi(s)=\infty$.
Analagously, if $\phi$~is a maximising query and $s\not\models\phi[-\infty]$,
then $\val_\phi(s)=-\infty$.

To conclude this section, let us remark that
queries using the reward-bounded \emph{release}
operator~$\Release$
can easily be accommodated in our framework.
For instance, the query
$\AllP_{\geq p}(a\Release_{\leq ?} b)$
is equivalent to the query
$\AllP_{\leq 1-p}(\neg a\Until_{\leq ?}\neg b)$.

\section{Evaluating qualitative queries}
\label{section:qualitative}

In this section, we give a strongly polynomial-time algorithm for evaluating
\emph{qualitative queries}, \ie queries where the probability bound~$p$
is either $0$ or~$1$.
Throughout this section, let
$\calM=(\States,\Actions,\gamma,\Labels,\reward,\delta)$
be an MDP with nonnegative rewards.
By \cref{prop:duality}, we can restrict to queries using one of
the quantifiers $\AllP_{>0}$, $\ExP_{>0}$, $\AllP_{=1}$ and $\ExP_{=1}$.
The following lemma allows to give a unified treatment of all
cases. ($\Next$ denotes the next-step operator).

\begin{lemma}\label{lemma:equivalences}
The equivalence
$Q\,\Next(a\Until(\neg a\wedge\psi))\equiv
Q\,\Next(a\Until(\neg a\wedge Q\,\psi))$
holds in $\PRCTL^*$ for all
$Q\in\{\AllP_{>0},\ExP_{>0},\AllP_{=1},\ExP_{=1}\}$,
$a\in\AP$, and all path formulas~$\psi$.
\end{lemma}

\Cref{alg:qual-until-upper} is our algorithm for
computing the values of a quantile query where we look for
an upper bound on the accumulated reward.
\begin{algorithm}
\begin{tabbing}
\hspace*{1em}\=\hspace{1em}\=\hspace{1em}\=\hspace{1em}\=
\hspace{1em}\=\hspace{1em}\= \kill
\emph{Input:} MDP~$\calM=(\States,\Actions,\gamma,\Labels,\reward,\delta)$,
$\phi=Q(a\Until_{\leq ?} b)$ \\[1ex]
\textbf{for each} $s\in\States$ \textbf{do} \\
\>\textbf{if} $s\models b$ \textbf{then}
 $v(s)\assign 0$ \textbf{else} $v(s)\assign\infty$ \\
$X\assign\{s\in\States:v(s)=0\}$; $R\assign\{0\}$\\
$Z\assign\{s\in\States:\text{$s\models a\wedge\neg b$ and $\reward(s)=0$}\}$ \\
\textbf{while} $R\neq\emptyset$ \textbf{do}\+\\
$r\assign\min R$\,; $Y\assign\{s\in X:v(s)\leq r\}\setminus Z$ \\
\textbf{for each} $s\in\States\setminus X$ with
$s\models a\wedge Q\,\Next(Z\Until Y)$ \textbf{do} \+\\
$v(s)\assign r+\reward(s)$ \\
$X\assign X\cup\{s\}$; $R\assign R\cup\{v(s)\}$ \-\\
$R\assign R\setminus\{r\}$ \-\\
\textbf{return}~$v$
\end{tabbing}
\vspace*{-2ex}
\caption{\label{alg:qual-until-upper}Solving qualitative queries of the form
$Q(a\Until_{\leq ?} b)$}
\end{algorithm}
The algorithm maintains a set $X$ of states, a
set~$R$ of real numbers, and a table~$v$ mapping states to
non-negative real numbers or infinity.
The algorithm works by discovering states with finite value
repeatedly until only the states with infinite value remain.
Whenever a new state is discovered, it is put into~$X$
and its value is put into~$R$.
In the initialisation phase, the algorithm discovers all states
labelled with~$b$, which have value~$0$.
In every iteration of the main loop, new states are discovered
by picking the least value~$r$ that has not been fully processed
(\ie the least element of~$R$) and checking which undiscovered $a$-labelled
states fulfil
the $\PCTL^*$ formula $Q\,\Next(Z\Until Y)$, where $Y$~is the set of already
discovered states whose value is at most~$r$ and $Z$~is the set of states
labelled with~$a$ but not with~$b$ and having reward~$0$. Any such
newly discovered state~$s$ must have value~$r+\reward(s)$, and
$r$~can be deleted from~$R$ at the end of the current iteration.
The termination of
the algorithm follows from the fact that in every iteration
of the main loop either the set~$X$ increases or it remains
constant and one element is removed from~$R$.

\begin{lemma}\label{lemma:qual-until-upper}
Let $\calM$ be an MDP, $\phi=Q(a\Until_{\leq ?} b)$
a qualitative query, and let $v$ be
the result of \cref{alg:qual-until-upper} on
$\calM$ and~$\phi$.
Then $v(s)=\val_\phi(s)$
for all states~$s$.
\end{lemma}

\begin{proof}
We first prove that $s\models\phi[v(s)]$ for all states~$s$
with $v(s)<\infty$.
Hence, $v$~is an upper bound on~$\val_\phi$. We prove this by
induction on the number of iterations the while loop has performed
before assigning a finite value to~$v(s)$. Note that this is the
same iteration when $s$~is put into~$X$ and that $v(s)$ never changes
afterwards. If $s$~is put into~$X$ before the
first iteration, then $s\models b$ and therefore also
$s\models\phi[0]=\phi[v(s)]$. Now assume
that the while loop has already completed $i$~iterations
and is about to add $s$ to~$X$ in the current iteration; let
$X$, $r$ and~$Y$ be as at the beginning
of this iteration (after $r$ and~$Y$ have been assigned,
but before any new state is added to~$X$). By the induction
hypothesis, $t\models\phi[r]$ for all $t\in Y$.
Since $s$~is added to~$X$, we have that $s\models a\wedge
Q\,\Next(Z\Until Y)$.
Using \cref{lemma:equivalences} and some basic $\PRCTL^*$ laws,
we can conclude that $s\models\phi[v(s)]$ as follows:
\begin{align*}
& s\models a\wedge Q\,\Next(Z\Until Y) \\
\Longrightarrow\quad &
s\models a\wedge Q\,\Next(Z\Until(\neg Z\wedge Q(a\Until_{\leq r} b))) \\
\Longrightarrow\quad &
s\models a\wedge Q\,\Next(Z\Until(\neg Z\wedge(a\Until_{\leq r} b))) \\
\Longrightarrow\quad &
s\models a\wedge Q\,\Next(a\Until_{\leq r} b) \\
\Longrightarrow\quad &
s\models Q(a\Until_{\leq r+\reward(s)} b) \\
\Longrightarrow\quad &
s\models\phi[v(s)]
\end{align*}

To complete the proof, we need to show that $v$~is also
a lower bound on~$\val_\phi$.
We define a strict partial order~$\prec$ on states by setting
$s\prec t$ if one of the following conditions holds:
\begin{enumerate}
\item $s\models b$ and $t\not\models b$,
\item $\val_\phi(s)<\val_\phi(t)$, or
\item $\val_\phi(s)=\val_\phi(t)$ and $\reward(s)>\reward(t)$.
\end{enumerate}
Towards a contradiction, assume that the set~$C$ of states~$s$
with $\val_\phi(s)<v(s)$ is non-empty, and pick a state $s\in C$
that is minimal with respect to~$\prec$ (in particular,
${\val_\phi(s)<\infty}$). Since $s\models\phi[\infty]$ and
the algorithm correctly sets $v(s)$ to~$0$ if $s\models b$,
we know that $s\models a\wedge\neg b$ and
$\val_\phi(s)\geq\reward(s)$. Moreover, by
\cref{prop:value}, $s\models\phi[\val_\phi(s)]$.
Let $T$ be the set of all states $t\in S\setminus Z$
such that $\val_\phi(t)+\reward(s)\leq\val_\phi(s)$,
\ie $t\models\phi[\val_\phi(s)-\reward(s)]$.
Note that $T\neq\emptyset$ (because every state labelled with~$b$
is in~$T$) and that $t\prec s$ for all $t\in T$.
Since $s$~is a minimal counter-example, we know that
$v(t)\leq\val_\phi(t)<\infty$ for all $t\in T$.
Consequently, after some number of iterations
of the while loop all elements of~$T$ have been added to~$X$
and the numbers~$v(t)$ have been added to~$R$. Since
$R$~is empty upon termination, in a following iteration we
have that $r=\max\{v(t):t\in T\}$ and that $T\subseteq Y$.
Let $x\coloneqq\val_\phi(s)-\reward(s)$.
Using \cref{lemma:equivalences} and some basic
$\PRCTL^*$ laws, we can conclude that
$s\models Q\,\Next(Z\Until Y)$ as follows:
\begin{align*}
& s\models\neg b\wedge\phi[\val_\phi(s)] \\
\Longrightarrow\quad &
s\models Q(\neg b\wedge(a\Until_{\leq x+\reward(s)} b)) \\
\Longrightarrow\quad &
s\models Q\,\Next(a\Until_{\leq x} b) \\
\Longrightarrow\quad &
s\models Q\,\Next(Z\Until(\neg Z\wedge(a\Until_{\leq x} b))) \\
\Longrightarrow\quad &
s\models Q\,\Next(Z\Until(\neg Z\wedge Q(a\Until_{\leq x} b)))
\displaybreak \\
\Longrightarrow\quad &
s\models Q\,\Next(Z\Until T) \\
\Longrightarrow\quad &
s\models Q\,\Next(Z\Until Y)
\end{align*}
Since also $s\models a$, this means that $s$~is added
to~$X$ no later than in the current iteration.
Hence, $v(s)\leq r+\reward(s)\leq\val_\phi(s)$,
which contradicts our assumption that $s\in C$.
\qed
\end{proof}

\begin{theorem}\label{thm:qual-until-upper}
Qualitative queries of the form $Q(a\Until_{\leq ?} b)$ can be
evaluated in strongly polynomial time.
\end{theorem}

\begin{proof}
By \cref{lemma:qual-until-upper}, \cref{alg:qual-until-upper}
can be used to compute the values of $Q(a\Until_{\leq ?} b)$.
During the execution of the algorithm,
the running time of one iteration of the while loop
is dominated by computing the set of states that fulfil
the $\PCTL^*$ formula $Q\,\Next(Z\Until Y)$, which can be
done in time
$\Oh(\abs{\delta})$ for $Q\in\{\AllP_{>0},\ExP_{>0},\AllP_{=1}\}$
and in time $\Oh(\abs{\States}\cdot\abs{\delta})$ for $Q=\ExP_{=1}$
(see \cite[Chapter 10]{BaierK08}).
In~each iteration
of the while loop, one element of~$R$ is removed, and the number
of elements that are put into~$R$ in total is bounded by the
number of states in the given MDP. Hence, the
number of iterations is also bounded by the number of states, and the
algorithm runs in time $\Oh(\abs{\States}\cdot\abs{\delta})$ or
$\Oh(\abs{\States}^2\cdot\abs{\delta})$, depending on~$Q$.
Finally, since the only arithmetic operation used by the algorithm is
addition, the algorithm is strongly polynomial.
\qed
\end{proof}

Of course, queries of the form $\ExP_{>0}(a\Until_{\leq ?} b)$
can actually be evaluated in time $\Oh(\abs{\States}^2+\abs{\delta})$
using D\ij{}kstra's algorithm since the value
of a state with respect to such a query is just the weight of a shortest path
from~$s$ via $a$-labeled states to a $b$-labelled state.

\Cref{alg:qual-until-upper} also gives us a useful
upper bound on the value of a state with respect to a
qualitative query.

\begin{proposition}\label{prop:bound-qual-until-upper}
Let $\calM$ be an MDP, $\phi=Q(a\Until_{\leq ?} b)$ a
qualitative quantile query, $n=\abs{\Labels^{-1}(a)}$, and
$c=\max\{\reward(s):s\in\Labels^{-1}(a)\}$. Then
$\val_\phi(s)\leq nc$ for all states~$s$ with
$\val_\phi(s)<\infty$.
\end{proposition}

\begin{proof}
By induction on the number of iterations \Cref{alg:qual-until-upper}
performs before assigning a finite number to $v(s)$.
\qed
\end{proof}

Finally, let us remark that our algorithm can be extended to handle queries
of the form $Q(a\Until_{>?} b)$, where a \emph{lower bound} on the
accumulated reward is sought.
To this end, the initialisation step has to be extended to identify
states with value~$-\infty$ and the rule for discovering new states has to be
modified slightly. We invite the reader to make the necessary modifications and
to verify the correctness of the resulting algorithm.
This proves that the fragment of \PRCTL with probability thresholds $0$
and~$1$ and without reward constraints of the form $=r$ can be model-checked
in polynomial time. Previously, a polynomial-time algorithm was only known
for the special case where the models are restricted to MDPs in which every
loop contains a state with nonzero reward \cite{JurdzinskiSL08}.

\section{Evaluating quantitative queries}
\label{section:quantitative}

In the following, we assume that all state rewards are natural numbers.
This does not limit the applicability of our results since any MDP~$\calM$
with nonnegative rational numbers as state rewards can be converted efficiently
to an MDP~$\calM'$ with natural rewards by multiplying all state rewards with
the least common multiple~$K$ of all denominators occurring in state rewards.
It follows that $\val^{\calM'}_\phi(s)=K\cdot\val^{\calM}_\phi(s)$
for any quantile query~$\phi$ and any state~$s$ of~$\calM$, so in order to
evaluate a quantile query on~$\calM$ we can evaluate it on~$\calM'$ and
divide by~$K$.
Throughout this section, we also assume that any transition probability and
any probability threshold~$p$ occurring in a quantile query is rational.
Finally, we define the \emph{size} of an MDP
$\calM=(\States,\Actions,\gamma,\Labels,\reward,\delta)$ to be
$\abs{M}\coloneqq\sum_{s\in\States}\size{\reward(s)}+
\sum_{(s,\alpha,t)\in\delta,\alpha\in\gamma(s)}\size{\delta(s,\alpha,t)}$,
where $\size{x}$~denotes the length of the binary representation of~$x$.

\subsection{Existential queries}
\label{sect:exists-until-upper}

In order to solve queries of the form $\ExP_{\geq p}(a\Until_{\leq ?} b)$ or
$\ExP_{>p}(a\Until_{\leq ?} b)$,
we first show how to compute the \emph{maximal} probabilities for fulfilling
the path formula $a\Until_{\leq r} b$ when we are given the reward bound~$r$.
Given an MDP~$\calM$, $a,b\in\AP$ and $r\in\bbN$, consider the following linear
program over the variables $x_{s,i}$ for $s\in\States$ and
$i\in\{0,1,\ldots,r\}$:
\begin{gather*}
\textstyle
\text{Minimise $\sum x_{s,i}$ subject to} \\
\begin{alignedat}{2}
x_{s,i} &\geq 0 &\quad\quad\quad\quad\quad\quad\quad\quad
 & \text{for all $s\in\States$ and $i\leq r$,} \\
x_{s,i} &= 1 && \text{for all $s\in\Labels^{-1}(b)$ and $i\leq r$,} \\
x_{s,i}
 &\geq\rlap{$\sum_{t\in\States}\delta(s,\alpha,t)\cdot x_{t,i-\reward(s)}$} \\
&&& \text{for all $s\in\Labels^{-1}(a)$, $\alpha\in\Actions$
 and $\reward(s)\leq i\leq r$.}
\end{alignedat}
\end{gather*}
This linear program is of size $r\cdot\abs{\calM}$, and it can be shown that
setting $x_{i,s}$ to $\max_{\sigma}\Prob_s^{\sigma}(a\Until_{\leq i} b)$
yields the optimal solution. Hence, we can compute the numbers
$\max_{\sigma}\Prob_s^{\sigma}(a\Until_{\leq i} b)$
in time $\poly(r\cdot\abs{\calM})$.
%

Our algorithm for computing the value of a state~$s$ \wrt a query of the
form $\ExP_{>p}(a\Until_{\leq ?} b)$ just computes
the numbers $\max_\sigma\Prob_s^{\sigma}(a\Until_{\leq i} b)$
for increasing~$i$ and stops as soon as this probability exceeds~$p$.
However, in order to make this algorithm work and to show that it does
not take too much time, we need a bound on the
value of~$s$ provided this value is not infinite.
Such a bound can be derived from the following lemma,
which resembles a result by Hansen~\ea, who gave a bound on the
convergence rate of \emph{value iteration} in
\emph{concurrent reachability games}~\cite{HansenIM11}.
Our proof is technically more involved though, since we have to deal
with paths that from some point onwards do not earn any more rewards.

\begin{lemma}\label{lemma:conv-max-prob-until-upper}
Let $\calM$ be an MDP where the denominator of each
transition probability is at most~$m$,
and let $n=\abs{\Labels^{-1}(a)}$,
$c=\max\{\reward(s):s\in\Labels^{-1}(a)\}$
and $r=k n c m^{-n}$ for some $k\in\bbN^+\!$.
Then
$\max_\sigma\Prob_s^\sigma(a\Until b)
<\max_\sigma\Prob_s^\sigma(a\Until_{\leq r} b)+\e^{-k}$
for all $s\in\States$.
\end{lemma}

\begin{proof}
Without loss of generality, assume that all $b$-labelled states
are absorbing.
Let us call a state~$s$ of~$\calM$ \emph{dead} if
$s\models\AllP_{=0}(a\Until b)$, and denote by~$D$ the set of
dead states. Note that $s\in D$ for all states~$s$
with $s\models\neg a\wedge\neg b$.
Finally, let $\tau$ be a memoryless scheduler such that
$\Prob_s^\tau({a\Until b})=\max_\sigma\Prob_s^\sigma(a\Until b)$
for all states~$s$, and denote by~$Z$ the set of all states~$s$
with $s\models a\wedge\neg b$ and $\reward(s)=0$.
By the definition of $D$ and~$Z$, we have that
$\Prob_s^\tau({a\Until_{\leq r}(D\vee\Globally Z)\wedge a\Until b})=0$
for all $s\in S$.
Moreover, if $s$~is not dead, then
there must be a simple path from~$s$ to a $b$-labelled state via
$a$-labelled states in the Markov chain induced by~$\tau$.
Since any $a$-labelled state has reward at most~$c$, this implies
that $\Prob_s^\tau(a\Until_{\leq nc} b)\geq m^{-n}$ for
all non-dead states~$s$.
Now let $\psi$ be the path formula $b\vee D\vee\Globally Z$.
We claim that
$\Prob_s^\tau(\neg(a\Until_{\leq r}\psi))<\e^{-k}$
for all states~$s$.
To prove this, let $s\in\States$.
We first show that
$\Prob_s^\tau(a\Until_{\leq i+nc}\psi\mid\neg(a\Until_{\leq i}\psi))
\geq m^{-n}$ for all $i\in\bbN$ with
$\Prob_s^\tau(a\Until_{\leq i}\psi)<1$.
Let $X$ be the set of sequences $xt\in S^*\cdot S$
such that $xt\in{\{s\in\States\setminus D:s\models a\wedge\neg b\}}^*$,
$\reward(x)\leq i$ and $\reward(xt)>i$. It is easy to see that
the set $\{xt\cdot\States^\omega:xt\in X\}$ is a
partition of the set of infinite sequences over~$S$ that
violate $a\Until_{\leq i}\psi$. Using the fact that
$\tau$~is memoryless, we can conclude that\goodbreak
\begin{align*}
& \Prob_s^\tau(a\Until_{\leq i+nc}\psi\mid\neg(a\Until_{\leq i}\psi)) \\
\geq\; & \Prob_s^\tau(a\Until_{\leq i+nc} b\mid\neg(a\Until_{\leq i}\psi)) \\
=\; & \Prob_s^\tau(a\Until_{\leq i+nc} b\cap X\cdot\States^\omega)/
 \Prob_s^\tau(X\cdot\States^\omega) \displaybreak[0] \\
=\; & \sum_{xt\in X}\Prob_s^\tau(a\Until_{\leq i+nc} b\cap
 xt\cdot S^\omega)/\Prob_s^\tau(X\cdot\States^\omega) \\
=\; & \sum_{xt\in X}\Prob_t^{\tau}(a\Until_{\leq i-\reward(x)+nc} b)\cdot
 \Prob_s^\tau(xt\cdot S^\omega)/\Prob_s^\tau(X\cdot\States^\omega) \\
\geq\; & \sum_{xt\in X}\Prob_t^{\tau}(a\Until_{\leq nc} b)\cdot
 \Prob_s^\tau(xt\cdot S^\omega)/\Prob_s^\tau(X\cdot\States^\omega) \\
\geq\; & \sum_{xt\in X} m^{-n}\cdot\Prob_s^\tau(xt\cdot S^\omega)/
 \Prob_s^\tau(X\cdot\States^\omega) \\
=\; & m^{-n}\,.
\end{align*}
Now, applying this inequality successively, we get that
$\Prob_s^\tau(\neg(a\Until_{\leq r}\psi))
\leq (1-m^{-n})^{\frac{r}{nc}}=(1-m^{-n})^{km^n}<\e^{-k}$.
Finally,
\begin{align*}
\Prob_s^\tau(a\Until b)
&=\Prob_s^\tau(a\Until b\wedge\neg(a\Until_{\leq r}(D\vee\Globally Z))) \\
&\leq\Prob_s^\tau(\neg(a\Until_{\leq r}(D\vee\Globally Z))) \\
&\leq\Prob_s^\tau(\neg(a\Until_{\leq r}\psi)\vee(a\Until_{\leq r} b)) \\
&\leq\Prob_s^\tau(\neg(a\Until_{\leq r}\psi))
 +\Prob_s^\tau(a\Until_{\leq r} b) \\
&<\e^{-k}+\max\nolimits_\sigma\Prob_s^\sigma(a\Until_{\leq r} b)
\end{align*}
for all $s\in\States$. Since
$\Prob_s^\tau(a\Until b)=\max_\sigma\Prob_s^\sigma(a\Until b)$,
this inequality proves the lemma.
\qed
\end{proof}

Given an MDP~$\calM$ and $a,b\in\AP$, we denote by~$\tilde{\calM}$ the MDP
that arises from~$\calM$ by performing the following transformation:
\begin{enumerate}
\item In each state~$s$, remove all actions~$\alpha$
with
$\sum_{t\in\States}\delta(s,\alpha,t)\cdot\max_\sigma\Prob_t^\sigma(a\Until b)
<\max_\sigma\Prob_s^\sigma(a\Until b)$ from the set $\gamma(s)$ of enabled
actions.
\item Label all states~$s$ such that $s\models\P_{=0}(a\Until b)$ with~$b$.
\end{enumerate}
The following
\ifapdx
lemma (proved in the appendix)
\else
lemma, whose proof is rather technical,
\fi
allows us to reduce the query
$\ExP_{\geq p}(a\Until_{\leq ?} b)$ to the qualitative query
$\ExP_{=1}(a\Until_{\leq ?} b)$ in the special case that $p$~equals the
optimal probability of fulfilling $a\Until b$.

\begin{lemma}\label{lemma:reduction-max-qual}
Let $\calM$ be an MDP, $\phi=\ExP_{\geq p}(a\Until_{\leq ?} b)$ and
$\tilde\phi=\ExP_{=1}(a\Until_{\leq ?} b)$.
Then $\val_\phi^{\calM}(s)=\val_{\tilde{\phi}}^{\tilde{\calM}}(s)$
for all states~$s$ of~$\calM$ with
$p=\max_\sigma\Prob_s^\sigma(a\Until b)$.
\end{lemma}

With the help of
\cref{lemma:conv-max-prob-until-upper,lemma:reduction-max-qual},
we can devise an upper bound for the value of any query
whose value is finite.

\begin{lemma}\label{lemma:value-exists-until-upper}
Let $\calM$ be an MDP
where the denominator of each transition
probability is at most~$m$,
$\phi=\ExP_{\rhd p}(a\Until_{\leq ?} b)$
for ${\rhd}\in\{\geq,>\}$,
$n=\abs{\Labels^{-1}(a)}$,
$c=\max\{\reward(s):s\in\Labels^{-1}(a)\}$,
$s\in\States$, and $q=\max_\sigma\Prob_s^\sigma(a\Until b)$.
Then at least one of the following statements holds:
\begin{enumerate}
\item $p\geq q$ and $\val_\phi(s)=\infty$.
\item $p=q$, ${\rhd}={\geq}$ and $\val_\phi(s)\leq nc$.
\item $p<q$ and $\val_\phi(s)\leq kncm^n$,
where $k=\max\{-\lfloor\ln(q-p)\rfloor,1\}$.
\end{enumerate}
\end{lemma}

\begin{proof}
Clearly, if either ${\rhd}={>}$ and
$p\geq q$ or ${\rhd}={\geq}$ and
$p>q$, then $\val_\phi(s)=\infty$,
and 1.\ holds.
Now assume that $p=q$ and ${\rhd}={\geq}$.
By \cref{lemma:reduction-max-qual}, we have that
$\val^{\calM}_\phi(s)=\val^{\tilde\calM}_{\tilde\phi}(s)$.
Hence, if $\val^{\tilde\calM}_{\tilde\phi}(s)=\infty$,
then 1.\ holds. On the other hand, if
$\val^{\tilde\calM}_{\tilde\phi}(s)<\infty$, then
\cref{prop:bound-qual-until-upper} gives us that
$\val^{\tilde\calM}_{\tilde\phi}(s)\leq nc$, and 2.\ holds.
Finally, if $p<q$, then let $r\coloneqq kncm^n$.
By \cref{lemma:conv-max-prob-until-upper},
we have that
$\max_\sigma\Prob_s(a\Until_{\leq r} b)
>q-e^{-k}
\geq q-e^{\lfloor\ln(q-p)\rfloor}
\geq q-(q-p)
=p$,
\ie $s\models\ExP_{\rhd p}(a\Until_{\leq r} b)$.
Hence, $\val_\phi(s)\leq r$, and 3.\ holds.
\qed
\end{proof}

It follows from \cref{lemma:value-exists-until-upper} that we can compute
the value of a state~$s$ \wrt a query~$\phi$ of the form
$\ExP_{>p}(a\Until_{\leq ?} b)$
as follows: First compute the maximal probability~$q$ of fulfilling
$a\Until b$ from~$s$, which can be done in polynomial time.
If $p\geq q$, we know that the value of~$s$
\wrt~$\phi$ must be infinite. Otherwise, $\val_\phi(s)\leq
r\coloneqq kncm^n$,
where $k=\max\{-\lfloor\ln(q-p)\rfloor,1\}$, and we can find
the least~$i$ such that $\max_\sigma\Prob_s^\sigma(a\Until_{\leq i} b)>p$
by computing $\max_\sigma\Prob_s^\sigma(a\Until_{\leq i} b)$ for all
$i\in\{0,1,\ldots,r\}$, which can be done in time
$\poly(r\cdot\abs{\calM})$.
Since $r$~is exponential in the number of states
of the given MDP~$\calM$, the running time of this algorithm is
exponential in the size of~$\calM$.
If $\phi$~is of the form $\ExP_{\geq p}(a\Until_{\leq ?} b)$,
the algorithm is similar, but in the case that $p=q$, we
compute $\max_\sigma\Prob_s^\sigma(a\Until_{\leq i} b)$ for all
$i\in\{0,1,\ldots,nc\}$ in order to determine whether the value is
infinite or one of these numbers~$i$.

\begin{theorem}
\label{thm:exists-until-upper}
Queries of the form $\ExP_{\geq p}(a\Until_{\leq ?} b)$ or
$\ExP_{>p}(a\Until_{\leq ?} b)$ can be evaluated in exponential time.
\end{theorem}

\subsection{Universal queries}
\label{sect:forall-until-upper}

In order to solve queries of the form
$\AllP_{>p}(a\Until_{\leq ?} b)$,
we first show how to compute the \emph{minimal} probabilities for fulfilling
the path formula $a\Until_{\leq r} b$ when we are given the reward bound~$r$.
Given an MDP~$\calM$, $a,b\in\AP$ and $r\in\bbN$, consider the following linear
program over the variables $x_{s,i}$ for $s\in\States$ and
$i\in\{0,1,\ldots,r\}$:
\begin{gather*}
\textstyle
\text{Maximise $\sum x_{s,i}$ subject to} \\
\begin{alignedat}{2}
x_{s,i} &\leq 1 &\quad\quad\quad\quad\quad\quad
 & \text{for all $s\in\States$ and $i\leq r$,} \\
x_{s,i} &= 0 && \text{for all $s\in\States$ with
 $s\not\models\AllP_{>0}(a\Until_{\leq i} b)$ and $i\leq r$,} \\
x_{s,i}
 &\leq\rlap{$\sum_{t\in\States}\delta(s,\alpha,t)\cdot x_{t,i-\reward(s)}$} \\
&&& \text{for all $s\in\States\setminus\Labels^{-1}(b)$, $\alpha\in\Actions$
 and $\reward(s)\leq i\leq r$.}
\end{alignedat}
\end{gather*}
This program is of size $r\cdot\abs{\calM}$, and it can be shown that
setting $x_{i,s}$ to $\min_{\sigma}\Prob_s^{\sigma}({a\Until_{\leq i} b})$
yields the optimal solution.
Since the set of states~$s$ with $s\models\AllP_{>0}(a\Until_{\leq i} b)$
can be
computed in polynomial time (\cref{thm:qual-until-upper}), this means
that we can compute the numbers
$\min_{\sigma}\Prob_s^{\sigma}({a\Until_{\leq i} b})$
in time $\poly(r\cdot\abs{\calM})$.
The following lemma is the analogue of
\cref{lemma:conv-max-prob-until-upper} for minimal
probabilities.

\begin{lemma}\label{lemma:conv-min-prob-until-upper}
Let $\calM$ be an MDP where the denominator of each
transition probability is at most~$m$,
and let $n=\abs{\Labels^{-1}(a)}$,
$c=\max\{\reward(s):s\in\Labels^{-1}(a)\}$
and $r=k n c m^{-n}$ for some $k\in\bbN^+\!$.
Then
$\min_\sigma\Prob_s^\sigma(a\Until b)
<\min_\sigma\Prob_s^\sigma(a\Until_{\leq r} b)+\e^{-k}$
for all $s\in\States$.
\end{lemma}

\begin{proof}
Without loss of generality, assume that all $b$-labelled states
are absorbing.
Let us call a state~$s$ of~$\calM$ \emph{dull} if
$s\models\ExP_{=0}(a\Until b)$, and denote by~$D$ the set of
dull states. Note that $s\in D$ for all states~$s$
with $s\models\neg a\wedge\neg b$.
If $s$~is not dull, then it is easy to see that,
for any scheduler~$\sigma$, the probability of reaching a $b$-labelled state
from~$s$ in \emph{at most $n$~steps} (while seeing only $a$-labelled states
before reaching a $b$-labelled state) is at least~$m^{-n}$.
Since any $a$-labelled state has reward at most~$c$, we get that
$\Prob_s^\sigma(a\Until_{\leq nc} b)\geq m^{-n}$
for all non-dull states~$s$ and all schedulers~$\sigma$.
In the following, denote by~$Z$ the set
$\{s\in\States:\text{$s\models a\wedge\neg b$ and $\reward(s)=0$}\}$,
and let $\psi$ be the path formula $b\vee D\vee\Globally Z$.
In the same way as in the proof of
\cref{lemma:conv-max-prob-until-upper}, we can infer that
$\Prob_s^\sigma(\neg(a\Until_{\leq r}\psi))<\e^{-k}$ for
all states~$s$ and all schedulers~$\sigma$.
Now fix a scheduler~$\tau$ that minimises
$\Prob_s^\tau(a\Until_{\leq r} b)$ for all $s\in S$
and a scheduler~$\sigma$
such that $\Prob_s^\sigma(a\Until b)=0$ for all $s\in D$.
From $\tau$ and~$\sigma$, we devise another scheduler~$\tau^*$
by setting
\[
\tau^*(x)=\begin{cases}
\tau(x) & \text{if $x\in(\States\setminus D)^*$,} \\
\sigma(x_2) & \text{if $x=x_1\cdot x_2$ where
$x_1\in(\States\setminus D)^*$ and $x_2\in D\cdot\States^*$.}
\end{cases}
\]
Note that
$\Prob_s^{\tau^*}\!(a\Until_{\leq r}(D\vee\Globally Z)\wedge a\Until b)=0$
and $\Prob_s^{\tau^*}\!(a\Until_{\leq r}(D\vee\Globally Z))=
\Prob_s^\tau(a\Until_{\leq r}(D\vee\Globally Z))$
for all $s\in\States$.
Hence,\goodbreak
\begin{align*}
\Prob_s^{\tau^*}\!(a\Until b)
&=\Prob_s^{\tau^*}\!(a\Until b\wedge\neg(a\Until_{\leq r} (D\vee\Globally Z))) \\
&\leq\Prob_s^{\tau^*}\!(\neg(a\Until_{\leq r}(D\vee\Globally Z))) \\
&=\Prob_s^\tau(\neg(a\Until_{\leq r}(D\vee\Globally Z))) \\
&\leq\Prob_s^\tau(\neg(a\Until_{\leq r}\psi)\vee(a\Until_{\leq r}b)) \\
&\leq\Prob_s^\tau(\neg(a\Until_{\leq r}\psi))+
\Prob_s^\tau(a\Until_{\leq r}b) \\
&<e^{-k}+\Prob_s^\tau(a\Until_{\leq r}b) \\
&=e^{-k}+\min\nolimits_\sigma\Prob_s^\sigma(a\Until_{\leq r}b)
\end{align*}
for all $s\in\States$. Since
$\min_\sigma\Prob_s^\sigma(a\Until b)\leq\Prob_s^{\tau^*}\!(a\Until b)$,
this inequality proves the lemma.
\qed
\end{proof}

With the help of \cref{lemma:conv-min-prob-until-upper}, we can devise
an upper bound for the value of a query of the
form $\AllP_{>p}(a\Until_{\leq ?} b)$
in case this value is finite.

\begin{lemma}\label{lemma:value-forall-until-upper}
Let $\calM$ be an MDP
where the denominator of each transition
probability is $\leq m$,
$\phi=\AllP_{>p}(a\Until_{\leq ?} b)$,
$n=\abs{\Labels^{-1}(a)}$,
$c=\max\{\reward(s):s\in\Labels^{-1}(a)\}$,
$s\in\States$, and $q=\min_\sigma\Prob_s^\sigma(a\Until b)$.
Then one of the following statements holds:
\begin{enumerate}
\item $p\geq q$ and $\val_\phi(s)=\infty$.
\item $p<q$ and $\val_\phi(s)\leq kncm^n$,
where $k=\max\{-\lfloor\ln(q-p)\rfloor,1\}$.
\end{enumerate}
\end{lemma}

\begin{proof}
Clearly, if $p\geq q$, then $\val_\phi(s)=\infty$,
and 1.\ holds.
On the other hand, if $p<q$, then let $r\coloneqq kncm^n$.
By \cref{lemma:conv-min-prob-until-upper},
we have that
$\min_\sigma\Prob_s(a\Until_{\leq r} b)
>q-e^{-k}
\geq q-e^{\lfloor\ln(q-p)\rfloor}
\geq q-(q-p)
=p$,
\ie $s\models\AllP_{>p}(a\Until_{\leq r} b)$.
Hence, $\val_\phi(s)\leq r$, and 3.\ holds.
\qed
\end{proof}

\iffalse
It follows from \cref{lemma:value-forall-until-upper} that we can compute
the value of a state~$s$ \wrt a query~$\phi$ of the form
$\AllP_{>p}(a\Until_{\leq ?} b)$
as follows: First compute the minimal probability~$q$ of fulfilling
$a\Until b$ from~$s$, which can be done in polynomial time.
If $p\geq q$, we know that the value of~$s$
\wrt~$\phi$ must be infinite. Otherwise, $\val_\phi(s)\leq r\coloneqq
kncm^n$,
where $k=\max\{-\lfloor\ln(q-p)\rfloor,1\}$, and we can find
the least~$i$ such that $\min_\sigma\Prob_s^\sigma(a\Until_{\leq i} b)>p$
by computing $\min_\sigma\Prob_s^\sigma(a\Until_{\leq i} b)$ for all
$i\in\{0,1,\ldots,r\}$, which can be done in time $\poly(r\cdot\abs{\calM})$.
Since $k$~is exponential in the number of states
of the given MDP~$\calM$, the running time of this algorithm is exponential
in the size of~$\calM$.
If $\phi$~is of the form $\AllP_{\geq p}(a\Until_{\leq ?} b)$,
the algorithm is similar, but in the case that $p=q$, we
have to compute $\min_\sigma\Prob_s^\sigma(a\Until_{\leq i} b)$ for all
$i\in\{0,1,\ldots,nc\}$ in order to determine whether the value is infinite
or one of these numbers~$i$.
\else
As in the last section, \cref{lemma:value-forall-until-upper}
can be used to derive an exponential
algorithm for computing the value of a state \wrt a query of the form
$\AllP_{>p}(a\Until_{\leq ?} b)$.
\fi

\begin{theorem}
\label{thm:forall-until-upper}
Queries of the form $\AllP_{>p}(a\Until_{\leq ?} b)$
can be evaluated in exponential time.
\end{theorem}

Regarding queries of the form $\AllP_{\geq p}(a\Until_{\leq ?} b)$,
we can compute the value of a state~$s$ whenever
the probability $\min_\sigma\Prob_s^\sigma(a\Until b)$ differs
from~$p$ using the same algorithm.
However, in the case that
$p=\min_\sigma\Prob_s^\sigma(a\Until b)$
it is not clear how to bound the value of~$s$.
As the following example shows, the analogous bound
of~$nc$ for existential queries from
\cref{lemma:value-exists-until-upper}
does not apply in this case.

\begin{example}
Consider the MDP depicted in \cref{fig:example},
where $\Actions=\{\flat,\natural\}$ and
$q\in{[0,1[}$ is an arbitrary probability.
\begin{figure}
\begin{tikzpicture}[->,x=1.5cm,y=1.5cm]
\node[state] (0) at (0,0) {$s_0$\nodepart{lower}$0$};
\node[state] (1) at (-1.5,0) {$s_1$\nodepart{lower}$0$};
\node[state] (2) at (-2.25,1) {$s_2$\nodepart{lower}$0$};
\node[state] (3) at (-3,0) {$s_3$\nodepart{lower}$0$};
\node[state] (4) at (1.5,0) {$s_4$\nodepart{lower}$1$};
\node[state] (5) at (3,0) {$s_5$\nodepart{lower}$0$};

\draw (0) to node[above] {$\flat, 1$} (1);
\draw (0) to node[above] {$\natural, 1$} (4);
\draw (1) to node[above right] {$\flat, \frac{1}{2}$} (2);
\draw (1) to node[above] {$\flat, \frac{1}{2}$} (3);
\draw[loop,out=45,in=135,looseness=6] (4) to node[above] {$\natural,q$} (4);
\draw (4) to node[above] {$\natural,1-q$} (5);
\end{tikzpicture}
\caption{\label{fig:example}An MDP with nonnegative rewards.}
\end{figure}
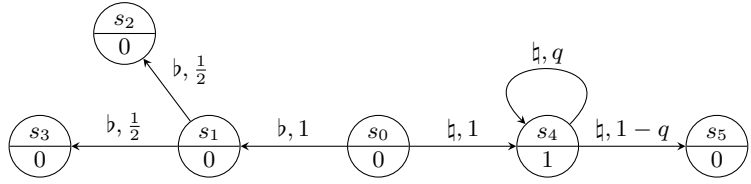
A state's reward is depicted in its bottom half, and a transition
from $s$ to~$t$ labelled with $\alpha,p$ indicates that
$\delta(s,\alpha,t)=p$. Only transitions from non-absorbing
states with nonzero probability and corresponding to enabled actions 
are shown. Assuming that every state is labelled with~$a$ but only
$s_3$ and~$s_5$ are labelled with~$b$, it is easy to see that
$\min_\sigma\Prob_{s_0}^\sigma(a\Until b)=\frac{1}{2}$.
Moreover, a quick calculation reveals
that the value of state~$s_0$ with respect to the query
$\AllP_{\geq 1/2}(a\Until_{\leq ?} b)$ equals
$-\lfloor 1/\log_2 q\rfloor$.
Since $q$~can be chosen arbitrarily close to~$1$,
this value can be made arbitrarily high.
\end{example}

\subsection{A pseudo-polynomial algorithm for Markov chains}

In this section, we give a \emph{pseudo-polynomial}
algorithm for evaluating quantile queries of
the form $\P_{\rhd\mkern 1mu p}(a\Until_{\leq ?} b)$
on Markov chains. (Note that the quantifiers $\ExP$ and
$\AllP$ coincide for Markov chains.)
More precisely, our algorithm runs in
time $\poly(c\cdot\abs{\calM}\cdot\size{p})$ if
$c$~is the largest reward in~$\calM$.
As an important
special case, our algorithm runs in polynomial time on
Markov chains where each state has reward $0$ or~$1$.

Our polynomial-time algorithm relies on the following equations
for computing the probability of the event $a\Until_{=i} b$ in
a Markov chain with rewards $0$ and~$1$.
Given such a Markov chain~$\calM$ and $a\in\AP$, we denote by~$Z$
the set of states~$s$ such that $\reward(s)=0$ and
$s\models a\wedge\neg b$.
Then the following equations hold for all $s\in\States$,
$a,b\in\AP$ and $r\in\bbN$:
\begin{itemize}
\item
$\Prob_s(a\Until_{=0} b)=\Prob_s(Z\Until b)$,
\item
$\Prob_s(a\Until_{=2r} b)=\sum_{t\in\States\setminus Z}
 \Prob_s(a\Until_{=r}\{t\})\cdot\Prob_t(a\Until_{=r} b)$,
\item
$\Prob_s(a\Until_{=2r+1} b)=\sum_{t\in\Labels^{-1}(a)\setminus Z}
 \sum_{u\in\States}
 \Prob_s(a\Until_{=r}\{t\})\cdot\delta(t,u)
 \cdot\Prob_u(a\Until_{=r} b)$,
\end{itemize}
Using these equations, we can compute the numbers
$\Prob_s(a\Until_{=r} b)$ along the binary representation
of~$r$ in time $\Oh(\poly(\abs{\calM})\cdot\log r)$ for
Markov chains with rewards $0$ and~$1$ (see also \cite{HanJon94}).
Since any Markov chain~$\calM$ with rewards $0,1,\ldots,c$
can easily be transformed into an equivalent Markov chain of
size $c\cdot\abs{\calM}$ with rewards $0$ and~$1$, the same numbers
can be computed in time
$\Oh(\poly(c\cdot\abs{\calM})\cdot\log r)$
for general Markov chains.
Finally, we can compute the numbers
$\Prob_s(a\Until_{\leq r} b)$ in the same time by first
applying the following operations to each $b$-labelled
state~$s$: Make $s$ absorbing, add~$a$ to $\Labels(s)$,
and set $\reward(s)=1$; in the resulting Markov chain
each state~$s$ fulfils
$\Prob_s(a\Until_{\leq r} b)=\Prob_s(a\Until_{=r} b)$.

Now let $\phi=\P_{\rhd\mkern 1mu p}(a\Until_{\leq ?} b)$.
Our algorithm for evaluating~$\phi$ at state~$s$ of
a Markov chain~$\calM$ is
essentially the same algorithm as for MDPs. Hence, we first
compute the probability $q\coloneqq\Prob_s(a\Until b)$.
If either $p>q$ or $p=q$ and ${\rhd}={>}$, then
$\val_\phi(s)=\infty$, by \cref{lemma:value-exists-until-upper}.
If $p<q$, then the same lemma entails that
$\val_\phi(s)\leq r\coloneqq kncm^n$, where
$n=\abs{\Labels^{-1}(a)}$, $m$ is the least
denominator of any transition probability, and
$k=\max\{-\lfloor\ln(q-p)\rfloor,1\}\leq\poly(\calM)+\size{p}$.
Hence, we can determine $\val_\phi(s)$ using an ordinary binary
search in time
$\Oh(\poly(c\cdot\abs{\calM})\cdot\log^2 r)=
\Oh(\poly(c\cdot\abs{\calM}\cdot\size{p}))$.
Finally, the same method can be applied if $p=q$ and ${\rhd}={\geq}$
since \cref{lemma:value-exists-until-upper} tells us that
$\val_\phi(s)\leq nc$ in this case.

\begin{theorem}
\label{thm:mc-until-upper}
Queries of the form $\P_{\geq p}(a\Until_{\leq ?} b)$ or
$\P_{>p}(a\Until_{\leq ?} b)$ can be evaluated in pseudo-polynomial
time on Markov chains.
\end{theorem}

\section{Conclusions}
\label{section:conclusion}

Although many researchers presented algorithms
and several sophisticated techniques for the \PCTL model checking
problem and to solve \PCTL and \PRCTL queries,
the class of quantile-based queries has not yet been addressed in
the model checking community.
In this paper, we presented algorithms for qualitative and 
quantitative quantile queries of the form
$\P_{\bowtie\mkern 1mu p}(a \Until_{\leq ?} b)$
and their duals
$\ExP_{\bowtie\mkern 1mu p}(a \Until_{\leq ?} b)$.
We established a polynomial algorithms for the qualitative case and
exponential algorithms for all but one of the quantitative cases.
Although the algorithms for the quantitative cases rely on a simple
search algorithm for the quantile, the crucial feature is the bound
we presented in
\cref{lemma:conv-min-prob-until-upper,lemma:conv-max-prob-until-upper}.
These bounds might be interesting also for other purposes.
There are several open problems to be studied in future work.
First, the precise complexity of quantitative quantile queries
is unknown and more efficient algorithms might exist, despite
the \NP-hardness shown in \cite{LaroussinieS05}.
Second, we concentrated here on reward-bounded until properties,
and by duality our results also apply to reward-bounded release properties.
But quantile queries can also be derived from other \PCTL-like formulas,
such as formulas reasoning about expected rewards,
\eg in combination with step bounds.

\paragraph*{Acknowledgments.}

We would like to thank Manuela Berg, Joachim Klein,
Sascha Kl\"uppelholz and Dominik Wojtczak for helpful
discussions and the anonymous reviewers for their
valuable remarks and suggestions.

\bibliographystyle{abbrv}
\bibliography{quantiles}

\ifapdx
\newpage
\appendix
\section{Proof of \cref{lemma:reduction-max-qual}}

In the following, we denote by~$D$ the set of states~$s$ of~$\calM$
such that
$s\models\AllP_{=0}(a\Until b)$ and assume without loss of generality
that each $b$-labelled state in~$\calM$ is absorbing.
Given a scheduler~$\sigma$ and sequence $x\in\States^*$, we also define
$\sigma[x]$ to be the scheduler such that $\sigma[x](y)=\sigma(xy)$
for all $y\in\States^*$. Finally, we write $\Finally a$ as an
abbreviation for the path formula $(\neg a\Until a)$.

Now let $s$ be a state of~$\calM$ such that
$p=\max_\sigma\Prob_s^\sigma(a\Until b)$.
Then it suffices to show that for all $r\in\bbN$ we have
$\calM,s\models\phi[r]$ if and only if
$\tilde\calM,s\models\tilde\phi[r]$.

$(\Rightarrow)$ Assume that $\calM,s\models\phi[r]$.
Hence, there exists a scheduler~$\sigma$ for~$\calM$ such that
$\Prob_s^\sigma(a\Until_{\leq r} b)=p$.
In particular, $\Prob_s^\sigma(a\Until b)=p$, which implies that
$\Prob_s^\sigma(a\Until(b\vee D))=1$.
We claim that
$\Prob_s^\sigma(a\Until_{\leq r}(b\vee D))
=\Prob_s^\sigma(a\Until(b\vee D))$.
Otherwise there would exist
$xt\in\States^*\cdot\States$ such that
$xt\in\{s\in\States\setminus D:s\models a\wedge\neg b\}^*$,
$\reward(xt)>r$ and $\Prob_s^\sigma(xt\cdot\States^\omega)>0$.
Since $t\not\in D$ and
$\Prob_s^\sigma(a\Until b)=\max_\tau\Prob_s^\tau(a\Until b)$,
we get that
$\Prob_t^{\sigma[x]}(a\Until b)>0$ and therefore also
$\Prob_s^\sigma(a\Until b)>\Prob_s^\sigma(a\Until_{\leq r} b)$,
a contradiction.
Finally, observe that $\sigma$~induces a scheduler~$\tilde\sigma$
for~$\tilde\calM$ such that
$\Prob_s^{\tilde\sigma}(a\Until_{\leq r} b)=
\Prob_s^{\tilde\sigma}(a\Until b)=1$, which proves
that $\tilde\calM,s\models\tilde\phi[r]$.

$(\Leftarrow)$ Assume that $\tilde\calM,s\models\tilde\phi[r]$.
Hence, there is a scheduler~$\tilde\sigma$ for~$\tilde\calM$ with
$\Prob_s^{\tilde\sigma}(a\Until_{\leq r} b)=1$.
This scheduler induces a scheduler~$\sigma$ for~$\calM$ such
that $\Prob_s^\sigma({a\Until_{\leq r}(b\vee D)})=1$.
Note that in~$\calM$ we have
$p=\max_\tau\Prob_s^\tau(\neg\Finally D)$.
(In particular, the memoryless, randomised scheduler~$\tau^*$ that in every
state~$t$ uniformly chooses an action from all those actions that maximise the
probability of staying in $\States\setminus D$ has the property that
$\Prob_s^{\tau^*}\!(a\Until b)=\Prob_s^{\tau^*}\!(\neg\Finally D)
=\max_\tau\Prob_s^\tau(\neg\Finally D)$.) Since $\sigma$~is derived
from a scheduler for~$\tilde\calM$, this implies that, from any state~$t$,
$\sigma$~never chooses an action that does not maximise the probability of
staying in $\States\setminus D$. But any such scheduler maximises the
probability of never reaching~$D$, \ie
$\Prob_s^\sigma(\Finally D)=1-\max_\tau\Prob_s^\tau(\neg\Finally D)
=1-p$. Hence,
$\Prob_s^\sigma(a\Until_{\leq r} b)
= 1-\Prob_s^\sigma(a\Until_{\leq r} D)
\geq 1-\Prob_s^\sigma(\Finally D)
=p$,
which proves that $\calM,s\models\phi[r]$.
\qed

\fi

\end{document}